\documentclass[11pt]{article}
\usepackage[T2A]{fontenc}
\usepackage[cp1251]{inputenc}
\usepackage[english]{babel}
\usepackage{amssymb}
\usepackage{amsmath}
\usepackage{amsthm}
\usepackage{amsfonts}
\usepackage{mathrsfs}
\usepackage{eufrak}

\usepackage[pdftex]{graphicx}
\usepackage{hyperref}

\usepackage{diagbox}

\newtheorem{theorem}{Theorem}

\newtheorem*{theorem*}{Theorem}

\newtheorem{lemma}{Lemma}

\newtheorem{proposition}{Proposition}

\newtheorem{definition}{Definition}

\newtheorem{example}{Example}
\newtheorem{remark}{Remark}

\begin{document}
\begin{center}
\textbf{\uppercase{L\'evy Laplacians, holonomy group and instantons on 4-manifolds}}

B.~O.~Volkov

\texttt{borisvolkov1986@gmail.com}

1) Steklov Mathematical Institute of Russian Academy of Sciences,

 ul. Gubkina 8, Moscow, 119991 Russia
 
2) Moscow Institute of Physics and Technology (National Research University),

Institutskiy per. 9, Dolgoprudny, Moscow Region, 141701  Russia 
\end{center}

{\bf Abstract:} 

The connection between Yang--Mills gauge fields on $4$-dimensional orientable compact Riemannian manifolds and modified L\'evy Laplacians is studied. A modified L\'evy Laplacian is obtained from the L\'evy Laplacian by the action of an infinite dimensional rotation. Under the assumption that the 4-manifold has a nontrivial restricted holonomy group of the bundle of self-dual 2-forms, the following is proved. There is a modified L\'evy Laplacian such that a parallel transport in some vector bundle over the 4-manifold is a solution of the  Laplace equation for this modified L\'evy  Laplacian if and only if the connection  corresponding to the parallel transport satisfies the Yang--Mills  self-duality (anti-self-duality)  equations.  An analogous connection between the Laplace equation for the L\'evy  Laplacian
and the Yang--Mills equations was previously known.

key words: L\'{e}vy Laplacian, infinite-dimensional Laplacians, Yang--Mills equations, instantons, holonomy group


\section*{Introduction}

The  L\'evy Laplacian is an infinite-dimensional  differential operator 
which has not a finite dimensional analog. The  equivalence  of the Yang--Mills equations and the Laplace equation for the L\'evy Laplacian is known (see~\cite{AGV1993,AGV1994,LV2001}). In this paper, we study a similar connection of a modification of this operator
with instantons. Instantons play an important role in the classical and quantum theory of gauge fields (see~\cite{BPST} and the survey paper~\cite{EGH}). An exposition of the mathematical theory of instantons can be found in~\cite{AHS,Taubes,ADHM,BL}.

Let $A$ be a connection in a vector bundle over a Riemannian manifold $M$ and 
 $F=dA+A\wedge A$ be  the associated curvature. The Yang--Mills action functional
 has the form
 \begin{equation*}
S_{YM}(A)=\frac 12\int_{M}\|F(x)\|^2Vol(dx).
\end{equation*}
The Euler--Lagrange equations for this action functional are the Yang--Mills equations:
\begin{equation}
\label{ymint}
D_A^\ast F=0,
\end{equation}
where  $D_A^\ast$ is the adjoint operator to the exterior covariant derivative generated
by $A$.
 Any connection $A$ generates the parallel transport $U^A$ which can been considered as a section in a vector bundle over
the Hilbert manifold of $H^1$-curves with the fixed origin  in $M$. 
If $f$ is a smooth section in this infinite-dimensional  vector bundle then the value of the L\'evy Laplacian $\Delta_L^{AGV}$ can been defined as
$$
\Delta_L^{AGV}f=\mathrm{tr}_L^{AGV}f'',
$$
where the L\'evy trace  $\mathrm{tr}_L^{AGV}$ is a special linear integral functional on some space of infinite-dimensional bilinear forms.  Also the L\'evy Laplacian   $\Delta_L^{AGV}$ can been defined as a Ces\`aro mean of the second-order directional derivatives (see~\cite{AS2006,Volkov2017,Volkov2019}). In this paper, this approach is not used. 
The L\'evy Laplacian $\Delta_L^{AGV}$  is a more complex analog of the  classical L\'evy Laplacian, which was defined by P. L\'evy for the functions on $L_2([0,1],\mathbb{R})$ (see~\cite{L1951}).

The L\'evy Laplacian $\Delta_L^{AGV}$ is connected with the Yang--Mills gauge fields in the following way.

{\it Theorem A.
The following two assertions are equivalent:
\begin{enumerate}
\item the  connection $A$ is a solution of the Yang--Mills  equations;
\item the parallel transport $U^A$  is a solution of the Laplace equation $\Delta^{AGV}_LU^A=0$.
\end{enumerate}
}

Theorem A was proved for $M=\mathbb{R}^d$ by Accardi, Gibilisco and Volovich in~\cite{AGV1994}
and generalized for the base Riemannian manifold by Leandre and Volovich in~\cite{LV2001}.

If $M$ is a oriented 4-manifold, it is possible to consider 
the Yang--Mills self-duality equations 
\begin{equation}
\label{dualintro}
F=\ast F
\end{equation} or the Yang--Mills  anti-self-duality   equations 
\begin{equation}
\label{antidualintro}
F=-\ast F
\end{equation}
 on a connection $A$, where $\ast$  is the Hodge star on the manifold $M$.
A connection
 is called  an instanton or  an anti-instanton if it is a solution of equations~(\ref{antidualintro}) or~(\ref{dualintro}) respectively.  The Bianchi identities $D_AF=0$ imply that instantons and
anti-instantons are solutions of the Yang--Mills equations~(\ref{ymint}).

It is a natural question whether the L\'evy Laplacian and instantons are related (see~\cite{Accardi}).
In~\cite{VolkovLLI} by author, the family of the modified L\'evy Laplacians was introduced which is connected with instantons and
anti-instantons. Unlike the usual Laplacian, the L\'evy Laplacian $\Delta_L^{AGV}$ is not rotation invariant and
any smooth curve $W\in C^1([0,1],SO(4))$ defines a modified L\'evy Laplacian $\Delta_L^W$ which acts 
on a smooth section $f$ in the infinite dimensional bundle by the formula:
$$
\Delta_L^Wf=\mathrm{tr}_L^{AGV}(W^\ast f''W).
$$

The Lie group  $SO(p)$ is not simple only if $p=4$.
In this case, where two normal subgroups $S^3_L$ and $S^3_R$ of $SO(4)$. The Lie algebra $so(4)$ can been decomposed as a direct sum of Lie algebras
$$so(4)=\mathrm{Lie}(S^3_L)\oplus \mathrm{Lie}(S^3_R),$$ where
$$\mathrm{Lie}(S^3_L)\cong \mathrm{Lie}(S^3_R)\cong so(3).$$ This decomposition corresponds to the decomposition of the bundle of 2-forms into the direct 
sum  of the sub-bundles of self-dual and anti-self-dual  2-forms.

The following theorem was proved for instantons
over $\mathbb{R}^4$ in~\cite{VolkovLLI}.

{\it Theorem B.
\label{LLandI}
Let $W\in C^1([0,1],S^3_L)$ ($W\in C^1([0,1],S^3_R)$) and
$$dim\, span \{W^{-1}(t)\dot{W}(t)\}_{t\in[0,1]}\geq2.$$
 Let the value of the Yang--Mills action functional $S_{YM}$ is finite on a connection $A$. 
 The following two assertions are equivalent:
\begin{enumerate}
\item  the connection $A$ on $\mathbb{R}^4$ is  an instanton (anti-instanton);
\item  the parallel transport  $U^A$ is a solution of  the  Laplace  equation for the modified L\'evy Laplacian $\Delta_L^{W}$:
\begin{equation*}
 \Delta_L^{W} U^A=0.
\end{equation*}
\end{enumerate}
}

Theorem B  means that under some technical assumptions the anti-self-duality equations  on $\mathbb{R}^4$ are equal to the Laplace equation  for the  modified L\'evy Laplacian. In the proof of Theorem B, the fact  that $\mathbb{R}^4$  is not compact was essentially used.
In~\cite{Volkov2020} by author, the following theorem  was proved for instantons
on an orientable Riemannian 4-manifold $M$.

{\it Theorem C.
Let $\{\bf{e_1},\bf{e_2},\bf{e_3}\}$ be  a basis of the Lie algebra $\mathrm{Lie}(S^3_L)$ (Lie algebra $\mathrm{Lie}(S^3_R)$). 
Let $W_i(t)=e^{t\bf{e_i}}$ for $i\in \{1,2,3\}$.
The following two assertions are equivalent:
\begin{enumerate}
\item   the connection $A$ on $M$ is  an instanton (anti-instanton);
\item  the parallel transport  $U^A$ is a solution of  three Laplace  equations for the modified L\'evy Laplacians $\Delta^{W_i}_LU^A=0$ for $i\in\{1,2,3\}$.
\end{enumerate}
}

Theorem C   means that under some technical assumptions the anti-self-duality Yang--Mills equations (which are the system of the three nonlinear differential equations of the first order) are equal to the system of the three Laplace equations for the different modified L\'evy Laplacians.

It turns out that, in the case of a Riemannian manifold, the situation is influenced by the holonomy group.
In this  paper,  we strengthen the results of~\cite{Volkov2020}.  We prove that an analog of Theorem B holds for orientable compact  Riemannian manifolds with 
nontrivial  restricted holonomy group $\mathrm{Hol}_m^0(\Lambda^2_+(T^\ast M))$ of the bundle of self-dual 2-forms. (Note that anti-self-dual Yang--Mills equations are conformally invariant and, in the case of the trivial restricted holonomy group $\mathrm{Hol}_m^0(\Lambda^2_+(T^\ast M))$, it is possible to consider a conformally equivalent
metric with a nontrivial holonomy group.)
 We find the sufficient conditions 
on the curve $W\in C^1([0,1],S^3_L)$  such that  the following two assertions are equivalent:
\begin{enumerate}
\item   the connection $A$  is  an instanton;
\item  the parallel transport  $U^A$ is a solution of  the Laplace  equation for the  L\'evy Laplacian $\Delta^{W}_LU^A=0$.
\end{enumerate}
These conditions are different for the cases when the holonomy group $\mathrm{Hol}_m^0(\Lambda^2_+(T^\ast M))$ coincides with
 $SO(2)$ or $SO(3)$.

 It is known from the works of Atiyah and Donaldson (see~\cite{Atiyah,Donaldson}) about a 
one-to-one  correspondence between the moduli space of instantons on
$\mathbb{R}^4$ and the space of  based holomorphic maps from $\mathbb C\mathrm P^1$ to the loop space.
(In~\cite{Sergreev}, a similar correspondence between  the moduli space of  Yang--Mills fields and  the space of based harmonic
maps from $\mathbb C\mathrm P^1$ to the loop space was conjectured.) 
 It means that there is a connection between the
Yang--Mills fields and chiral fields with finite dimensional base manifolds and infinite dimensional target
spaces. On the contrary, a parallel transport generated by the Yang--Mills field can been considered as a chiral field
with an infinite-dimensional base manifold and a finite-dimensional
target space (see~\cite{Polyakov1979,Polyakov1980,Polyakov1987}).  In~\cite{AV1981},  the equations of the motion of the chiral fields on the parallel transport  with the divergence associated with the L\'evy Laplacian were considered (see also~\cite{Volkov2017,Volkov2019}).  The approach to the Yang--Mills fields based on the L\'evy Laplacian goes back to the paper~\cite{AV1981}. 
 Different approaches to the Yang--Mills fields  based on the parallel transport but not based on the L\'evy Laplacian
were used in~\cite{Gross,Driver,Bauer1998,Bauer2003,ABT,AT}. For a recent development in the study of the L\'evy Laplacian in the white noise theory, see~\cite{Volkov2018,AHJS}.

The paper is organized as follows. 
In Sec.~\ref{Sec1}, we give preliminary information about the Yang--Mills equations and instantons on  4-dimensional orientable Riemannian   manifolds.  In Sec.~\ref{Sec2}, we give  preliminary information  about the parallel transport and the holonomy group of the bundle of self-dual 2-forms.  
In Sec.~\ref{Sec3}, we give the scheme of  the definition of the second order differential operators on  the space of sections in an  infinite dimensional  vector bundle. Using this scheme, we define the family of the modified L\'evy Laplacians  parameterized  by the choice of a curve in the group $SO(4)$. We give the value of the modified L\'evy Laplacian on the parallel transport.  In Sec.~\ref{Sec4}, we  formulate and prove the main theorem
on the equivalence of the anti-self-duality   Yang--Mills equations on a manifold with a nontrivial holonomy group of the bundle of self-dual 2-forms and 
the Laplace equation for some modified L\'evy Laplacian.

\section{Instantons on 4-manifold}
\label{Sec1}

In this section, we give geometric preliminaries about Yang--Mills connections and instantons.
For more information see, for example,~\cite{FU,BL}.

Let $M$ be  a smooth  orientable   compact $4$-dimensional
Riemannian  manifold with a metric $g$. We will use the Einstein's notation for summation over the repeated indices and  will raise and lower indices using this  metric $g$. Let $G$ be a closed Lie group  realized as a subgroup of $SO(N)$.
The symbol $\mathfrak{g}$ denotes the Lie algebra of  $G$ endowed by the trace product:
$$(B_1,B_2)_{\mathfrak{g}}=-\mathrm{tr}(B_1B _2),\;\;B_1,B_2\in \mathfrak{g}.$$
Let $E=E(\mathbb{R}^N,\pi,M,G)$  be a vector bundle over $M$ with the
projection $\pi\colon E\to M$ and the structure group  $G$.  The fiber over $x\in M$ is $E_x=\pi^{-1}(x)\cong\mathbb{R}^N$.
Let $P$ be the principle bundle  over
$M$  associated with $E$ and  $\mathrm{ad} (P)=\mathfrak{g}\times_G M$   be the
adjoint bundle of $P$.
A connection $A(x)=A_\mu(x)dx^\mu$ in the vector bundle $E$ is  a smooth section in
$\Lambda^1(T^\ast M)\otimes \mathrm{ad}P$.  Let $\nabla$ denotes the covariant derivative generalized by this connection.
If $\phi$ is a smooth section in $\mathrm{ad}P$ then in local coordinates we have
$$
\nabla_\mu\phi=\partial_\mu \phi+[A_\mu,\phi].
$$
The curvature $F(x)=\sum_{\mu<\nu}F_{\mu\nu}(x)dx^\mu\wedge dx^\nu$ of the connection $A$ is a smooth section in $\Lambda^2(T^\ast M)\otimes \mathrm{ad}P$ defined by the formula $F_{\mu \nu}=\partial_\mu A_\nu-\partial_\nu A_\mu+[A_\mu,A_\nu]$.

Let $$D_A\colon C^\infty(M,\Lambda^p(T^\ast M)\otimes \mathrm{ad}P)\to C^\infty(M,\Lambda^{p+1}(T^\ast M)\otimes \mathrm{ad}P)$$ be the operator of the exterior covariant derivative. 
It is defined by its action on  forms $\alpha\otimes \phi$, where $\alpha$ is  a real $p$-form and
$\phi$ is a section in $\mathrm{ad}P$, by the formula
$$
D_A(\alpha\otimes \phi)=d\alpha\otimes \phi+(-1)^{p}\alpha\otimes \nabla \phi,
$$
where $d$ denotes the operator of the usual exterior derivative.
Let $$D_A^{\ast}:C^\infty(M,\Lambda^{p+1}(T^\ast M)\otimes \mathrm{ad}P) \to C^\infty(M,\Lambda^{p}(T^\ast M)\otimes \mathrm{ad}P)$$ be a formally adjoint to the operator $D_A$.
We have $D_A^{\ast}=-\ast D_A\ast$, where $\ast$ is the Hodge star on the manifold $M$.

The Yang--Mills action functional has the form
\begin{equation}
\label{YMaction1}
S_{YM}(A)=\frac 12\int_{M}\|F(x)\|^2Vol(dx),
\end{equation}
where $Vol$ is the Riemannian volume measure on the manifold $M$ and $\|\cdot\|$ is a norm generated by the trace product on the Lie algebra $\mathfrak{g}$.
The Yang--Mills equations  on a connection $A$ have the form
\begin{equation}
\label{YMequations}
(D_A^\ast F)=0.
\end{equation}
In local coordinates, we have
$$
(D_A^\ast F)_\nu=-\nabla^\mu F_{\mu\nu}
$$
and
 \begin{equation*}
 \nabla_\lambda F_{\mu
\nu}=\partial_\lambda F_{\mu \nu}+[A_\lambda,F_{\mu
\nu}]-F_{\mu
\kappa}\Gamma^\kappa_{\lambda\nu}-F_{\kappa\nu}\Gamma^\kappa_{\lambda\mu},
\end{equation*}
where $\Gamma^\kappa_{\lambda\nu}$  are the  Christoffel symbols of the L\'evy-Civita connection on $M$.
Solutions of the Yang--Mills equations are critical points of the Yang--Mills action functional~(\ref{YMaction1}).
These critical points are called Yang--Mills connections.

 Let $F_-=\frac 12 (F-\ast F)$ and $F_+=\frac 12 (F+\ast F)$ be the  anti-self-dual and self-dual parts of the curvature $F$ respectively.
 (We have $\ast F_-=-F_-$ and $\ast F_+=F_+$). 
 The connection $A$ is called anti-instanton (instanton)  if it is a solution of
the self-duality (anti-self-duality) Yang--Mills equations:
  \begin{equation} 
\label{autodual1}
F_-=0\,(F_+=0).
\end{equation}
The Yang--Mills action functional can been rewritten in the form
$$
S_{YM}(A)=\frac 12\int_{M}(\|F_-(x)\|^2+\|F_+(x)\|^2)Vol(dx).
$$ 
The Pontryagin number or topological charge
$$
k(A)=\frac 1{8\pi^2}\int_{M}(\|F_+(x)\|^2-\|F_-(x)\|^2)Vol(dx)
$$
is an integer topological invariant of the vector bundle. The inequality
$$
S_{YM}(A)\geq 4\pi^2|k(A)|
$$
imply that the instantons and  the anti-instantons, if they are exist, are local extrema of this  functional.

The gauge transform is a smooth section in $\mathrm{Aut} P$. Such a section $\psi$ acts on the connection by the formula
\begin{equation}
A\to A'=\psi^{-1}A\psi+\psi^{-1}d\psi
\end{equation}
and on the curvature by the formula
\begin{equation}
F\to F'=\psi^{-1}F\psi.
\end{equation}
The Lagrange function of~(\ref{YMaction1}), the Yang--Mills equations~(\ref{YMequations}), the self-duality equations and the anti-self-duality equations are invariant under the action of gauge transform.

The Weyl transformation (conformal transformation) of the metric $g\mapsto\overline{g}$ is defined by the formula 
$\overline{g}_{\mu\nu}=e^{2\varphi} g_{\mu\nu}$,
where $\varphi$ is smooth function on $M$.
In the dimension four, the Yang--Mills functional~(\ref{YMaction1}), instantons and anti-instantons are invariant under the 
Weyl transformation.   

Let $W_+$ and $W_-$ be the self-dual and anti-self-dual parts of the Weyl tensor (the conformally invariant part
of the Riemann curvature tensor) respectively. 
An oriented Riemannian 4-manifold  is called self-dual or anti-self-dual if $W_-=0$ or $W_+=0$ respectively.
It was proved in~\cite{AHS} by Atiyah, Hitchin and Singer, that in the case of the self-dual base manifold, 
the moduli space of  instantons (the factor space of all instantons with the respect to the gauge equivalence)  is a non-empty finite dimensional manifold.

Examples of self-dual manifolds (see~\cite{AHS}):
\begin{itemize}
\item Conformally flat manifolds. For example, 4-sphere $S^4$ or $S^1\times S^3$.
\item The complex projective plane  $\mathbb C\mathrm P^2$ is self-dual.
\item The 4-dimensional Calabi--Yau manifold is anti-self-dual with respect to it canonical orientation.
\end{itemize}

 If the intersection form on the manifold is indefinite then there are not instantons and anti-instantons on this manifold.
 The result by Atiyah, Hitchin and Singer was generalized  for manifolds with the positive intersection form by  Taubes in~\cite{Taubes}
 (see also~\cite{FU}). 
The moduli space of instantons on  $\mathbb{R}^4$  was described by Atiyah, Drinfeld, Hitchin and Manin in~\cite{ADHM}.
In this paper, we show that theorem on equivalence of the anti-self-duality equations and the Laplace equation for the modified L\'evy Laplacian  is  true for  anti-self-dual compact manifolds which are not  Calabi--Yau.

\section{Parallel transport and holonomy group}
\label{Sec2}

In this section, we give geometric preliminaries about the parallel transport and the holonomy group.

\subsection{Parallel transport}
\label{Sec2a}

For any interval $I$ let  the symbol $H^1(I,\mathbb{R}^4)$ denote the Sobolev space of $\mathbb{R}^4$-valued 
functions on $I$.
 It is the Hilbert space with scalar product
   $$(h_1,h_2)_1=\int_I(h_1(t),h_2(t))_{\mathbb{R}^4}dt+\int_I(\dot{h}_1(t),\dot{h}_2(t))_{\mathbb{R}^4}dt.$$
Let $\gamma$ be the mapping from $[0,1]$ to $M$.  It is called $H^1$-curve if $\phi_a\circ \gamma\mid_I\in H^1(I,\mathbb{R}^4)$
 for any  sub-interval $I$, such that  $\gamma(I)\subset W_a$. Here $(\phi_a, W_a)$  is
some coordinate chart of the manifold $M$.
Let $\Omega_m$ denote the set  of the $H^1$-curves with the origin at $m\in M$. 
This set can been endowed with the natural structure of  a Hilbert manifold modeled on the Hilbert space $H^1_0([0,1],\mathbb{R}^4)=\{\gamma\in H^1([0,1],\mathbb{R}^4)\colon \gamma(0)=0\}$ (see~\cite{Klingenberg,Klingenberg2,Driver}).

 For any  $H^1$-curve  $\gamma$  an operator $U^A_{t,s}(\gamma)\in \mathrm{Hom}(E_{\gamma(s)},E_{\gamma(t)})$, where $0\leq s\leq t \leq1$, is a solution  of the system 
 \begin{equation}
\label{form 2}
 \left\{
\begin{aligned}
\frac
{d}{dt}U^A_{t,s}(\gamma)=-A_{\mu}(\gamma(t))\dot{\gamma}^{\mu}(t) U^A_{t,s}(\gamma)\\
\frac
d{ds}U^A_{t,s}(\gamma)=U^A_{t,s}(\gamma)A_{\mu}(\gamma(s))\dot{\gamma}^{\mu}(s)\\
\left.U^A_{t,s}(\gamma)\right|_{t=s}=I_N.
\end{aligned}
\right.
\end{equation}
The operator $U^A_{1,0}\in \mathrm{Hom}(E_{\gamma(0)},E_{\gamma(1)})$ is a parallel transport along the curve $\gamma$ generated by the connection $A$.
Let  $\mathcal{E}_m$ be a Hilbert vector bundle over $\Omega_m$ such that its fiber over  $\gamma\in \Omega_m$ is the space  $\mathrm{Hom}(E_m,E_{\gamma(1)})$.
The mapping $\Omega_m\colon\gamma\to  U^A_{1,0}(\gamma)$ is a smooth section in this  vector bundle (see~\cite{Gross,Driver}).
Let $Q_{t,s}(\gamma)$  be a parallel transport with respect to the Levi-Civita connection in the tangent bundle $TM$ along  the restriction of the $H^1$-curve $\gamma$
on $[s,t]$.

\subsection{Holonomy group}
\label{Sec2b}

For  any $x\in M$ let $\Omega_{m,x}=\{\gamma\in \Omega_m\colon \gamma(1)=x\}$. Then $\Omega_{m,m}$ is the space of
$H^1$-loops based at $m$. Parallel transports with respect to the Levi-Civita connection  along  loops based at $m$ generate the holonomy group:  
$$
\mathrm{Hol}_m(M)=\{Q_{1,0}(\gamma)\colon \gamma\in \Omega_{m,m}\}.
$$
The restricted holonomy group  $\mathrm{Hol}^0_m(M)$ based at $m$  is the subgroup of $\mathrm{Hol}_m(M)$  generated by the parallel transports along contractible loops $\gamma\in\Omega_{m,m}$. Let $\mathrm{Hol}^0_x(M)$  denote the restricted holonomy group at $x\in M$ ($\mathrm{Hol}^0_x(M)\cong \mathrm{Hol}^0_m(M)$).
Due to $M$ is orientable and compact, $\mathrm{Hol}^0_x(M)$ is a connected closed Lie subgroup of $SO(4)$.

The Levi-Civita connection induces  the Riemannian connection on the vector bundle $\Lambda^2(TM)=\Lambda^2_+(TM)\oplus \Lambda^2_-(TM)$.
The group $\mathrm{Hol}^0_x(\Lambda^2(TM))=\{K\in \mathrm{Aut}(\Lambda^2(T_xM)) \colon K=H\wedge H, H\in \mathrm{Hol}^0_x(M)\}$ is the restricted holonomy group in 
$\Lambda^2(TM)$ at $x$.
The spaces $\Lambda^2_+(T_xM)$ and $\Lambda^2_-(T_xM)$ are invariant under the action of $K\in \mathrm{Hol}^0_x(\Lambda^2(TM))$.
So  $$\mathrm{Hol}^0_x(\Lambda^2(TM))=\mathrm{Hol}^0_x(\Lambda^2_+(TM))\times \mathrm{Hol}^0_x(\Lambda^2_-(TM)).$$ 
The group $\mathrm{Hol}^0_x(\Lambda^2_+(TM))$ is a closed connected  Lie subgroup of $SO(3)$.
   There are three possible cases:

\begin{itemize}
\item $\mathrm{Hol}_m^0(\Lambda^2_+(TM))$ is trivial. We can consider the following example.  If $M$ is a simply connected manifold, 
then holonomy groups and restricted holonomy groups coincide. 
If a simply connected  base manifold $M$ is irreducible and 
$\mathrm{Hol}_m(M)\cong SU(2)$, then  due to Berger's classification theorem $M$ (see~\cite{Besse}) is a Calabi--Yau manifold. 
In this case, $\mathrm{Hol}_m(\Lambda^2_+(TM))$ is trivial. 
 
\item $\mathrm{Hol}_m^0(\Lambda^2_+(TM))\cong SO(2)$. We can consider the following example.  If a simply connected 
base manifold $M$ is irreducible  and 
$\mathrm{Hol}_m(M)\cong U(2)$,  then  due to Berger's and Cartan's classification theorems  $M$ is Kahler. (If $M$ is symmetric, then  $M=\mathbb C\mathrm P^2$).  In this case, then  $\mathrm{Hol}_m(\Lambda^2_+(TM))\cong SO(2)$.

\item   $\mathrm{Hol}_m^0(\Lambda^2_+(TM))\cong SO(3)$.  One can show that if $\mathrm{Hol}^0_m(M)\cong SO(4)$ or  $\mathrm{Hol}^0_m(M)\cong SO(3)$, then 
$\mathrm{Hol}^0_m(\Lambda^2_+(TM))\cong SO(3)$.  For example, due to Berger's classification theorem in  a general case of the simply connected orientable  irreducible  nonsymmetric manifold $M$ we have $\mathrm{Hol}_m(M)\cong SO(4)$. If $M=S^4$ then also $\mathrm{Hol}_m(M)\cong SO(4)$.
The case $\mathrm{Hol}_m(M)\cong SO(3)$ can been realized if $M$ is reducible (for example, if $M=S^1\times S^3$).
\end{itemize}

\section{L\'evy Laplacian}
\label{Sec3}

In this section, we provide the general scheme of the definition of the 
second order differential operator. We provide the   definitions of the L\'evy Laplacian and the modified L\'evy Laplacians. We give the value of the modified 
L\'evy Laplacian on the parallel transport.

\subsection{Second order differential operators}
\label{Sec3a}

Fix an orthonormal  basis $\{e_1,e_2,e_3,e_4\}$  in the tangent space $T_mM$. 
We identify
$$
H^1_0([0,1],\mathbb{R}^4)\ni h(\cdot)=(h^\mu(\cdot))\leftrightarrow e_\mu h^\mu(\cdot)\in H^1_0([0,1],T_mM).
$$

Let  $\mathcal H^1_0$ be the tangent bundle over $\Omega_m$. Its  fiber over $\gamma\in \Omega_m$ is a Hilbert space $H^1_\gamma(TM)$  of 
$H^1$-fields ($H^1$-sections in the pullback bundle $\gamma^{*}TM$) $X$ along $\gamma$ such that $X(0)=0$ (see~\cite{Klingenberg,Klingenberg2,Driver}).
The scalar product 
on this space is defined by the formula
\begin{equation}
\label{metrG1}
G_1(X,Y)=\int_0^1g(X(t),Y(t))dt
\\+\int_0^1g(\nabla X(t),\nabla Y(t))dt,
\end{equation}
where
\begin{equation}
\label{covder}
\nabla X^\mu(t)= \dot{X}^\mu(t)+\Gamma^\mu_{\lambda \nu}(\gamma(t))X^\lambda(t)\dot{\gamma}^\nu(t).
\end{equation}

For any $\gamma\in \Omega_m$ the L\'evy-Civita connection generates the canonical isometric isomorphism  between
$H^1_0([0,1],\mathbb{R}^4)$ and $H^1_\gamma(TM)$, which action on $h\in H^1_0([0,1],\mathbb{R}^4)$ we will denote by $\widetilde{h}$. This isomorphism acts   by the formula
\begin{equation}
\label{isomorhism}
\widetilde{h}(\gamma;t)=Q_{t,0}(\gamma)h(t)=e_\mu(\gamma,t)h^\mu(t),
\end{equation}
where $e_\mu(\gamma,\cdot)=Q_{\cdot,0}(\gamma)e_\mu$ is the parallel transport of $e_\mu$ ($\mu\in\{1,2,3,4\}$) along $\gamma\in \Omega_m$
by the Levi-Civita connection. The space $H^1_\gamma(TM)$ is the tangent space to the Hilbert manifold $\Omega_m$ at $\gamma$. Let 
 $\mathcal H^1_0$ denote the tangent bundle over  $\Omega_m$.
Due to isomorphism~(\ref{isomorhism}), the tangent bundle $\mathcal H^1_0$ is trivial. 

 Let $\mathcal H^1_{0,0}$ denote the sub-bundle of $\mathcal H^1_0$ such that
 the  fiber of  $\mathcal H^1_{0,0}$  over  $\gamma\in\Omega_m$ is the space   $\{X\in H^1_\gamma(TM)\colon X(1)=0\}$.
By  canonical isomorphism~(\ref{isomorhism}) its fiber is isomorphic to the Hilbert space $H^1_{0,0}:=H^1_{0,0}([0,1],\mathbb{R}^4)=\{\gamma\in H^1_0([0,1],\mathbb{R}^4)\colon \gamma(1)=0\}$.
 
Let $X$ be a smooth section in $\mathcal H^1_{0,0}$.
If   $f$ is a smooth section in $\mathcal E_m$, then the derivative $d_Xf$ of $f$ along the field $X$ is correctly defined.
Due to triviality of the vector bundle $\mathcal H^1_{0,0}$, for any smooth section $f$ in  $\mathcal E_m$ there exists the section  $\widetilde{D}f$ in  $\mathcal E_m \otimes H^1_{0,0}\cong \mathcal H^1_{0,0}$ such that
$$d_{\widetilde{h}}f(\gamma)=<\widetilde{D}f(\gamma),h>$$ for any $h\in H^1_{0,0}$. Let $\mathcal{L}(H^1_{0,0},H^1_{0,0})$ denote the space of all continuous linear operators in $H^1_{0,0}$.  There exists the section  $\widetilde{D}^2f$ in $\mathcal E_m \otimes \mathcal{L}(H^1_{0,0},H^1_{0,0})$ such that
$$<d_{\widetilde{u}}\widetilde{D}f(\gamma),v>=<\widetilde{D}^2f(\gamma)u,v>$$ for any $u,v\in H^1_{0,0}$. 

\begin{definition}
Let $S$ be  a linear functional  on $domS\subset  \mathcal{L}(H^1_{0,0}, H^1_{0,0})$.
The domain of the second order differential operator $D^{2,S}$ associated with $S$   is the space of  all smooth sections $f$ in $\mathcal E_m$ such that $\widetilde{D}^2f(\gamma)\in domS$ for all $\gamma\in \Omega_m$. The second order differential operator $D^{2,S}$ 
 acts
on $f$ by the formula
$$
D^{2,S}f(\gamma)=S(\widetilde{D}^2f(\gamma)).
$$
\end{definition}

\subsection{Modified L\'evy Laplacians}
\label{Sec4b}

Let $T^2_{AGV}$  be the space of all continuous bilinear real-valued functionals on    $H_{0,0}^1\times H_{0,0}^1$ that have the form
\begin{multline}
\label{f''}
Q(u,v)=\int_0^1\int_0^1Q^V(t,s)<u(t),v(s)>dtds+\\
+\int_0^1Q^L(t)<u(t),v(t)>dt+
\\
+\frac 12\int_0^1Q^S(t)<\dot{u}(t),v(t)>dt+\frac 12\int_0^1Q^S(t)<\dot{v}(t),u(t)>dt,\;\;u,v\in H^1_{0,0},
\end{multline}
where
$Q^V\in L_2([0,1]\times[0,1],T^2(\mathbb{R}^4))$,
$Q^L\in L_1([0,1],Sym^2(\mathbb{R}^4))$,
$Q^S\in L_\infty([0,1],\Lambda^2(\mathbb{R}^4))$.
Here  $T^2(\mathbb{R}^4)$, $Sym^2(\mathbb{R}^4)$ and $\Lambda^2(\mathbb{R}^4))$ denote the spaces of all
tensors, all symmetrical tensors  and all antisymmetrical tensors of type $(0,2)$ on $\mathbb{R}^4$ respectively.
 The functions $Q^V(\cdot,\cdot)$,
 $Q^L(\cdot)$  and $Q^S(\cdot)$ are called the Volterra integral kernel, the L\'evy integral kernel and the singular integral kernel respectively. It was proved in~\cite{AGV1994} that these kernels are defined in a unique way (see also~\cite{Volkov2019}).

\begin{definition}
The L\'evy trace  is a linear functional  $\mathrm{tr}^{AGV}_L$ acting on $Q\in T^2_{AGV}$  by
$$
\mathrm{tr}^{AGV}_LQ=\int_0^1\mathrm{tr}\,Q^L(t)dt.
$$
The  L\'evy Laplacian $\Delta^{AGV}_L$  is the second order differential operator $D^{2,\mathrm{tr}^{AGV}_L}$.
\end{definition}

This operator was introduced by Accardi, Gibilisco and Volovich in~\cite{AGV1994} for the flat case and  by Leandre and Volovich in~\cite{LV2001} for the case of Riemannian manifold.

Let us introduce the modification of the L\'evy trace that is connected with instantons. 
Let $W\in C^1([0,1],SO(4))$.
We can consider $W$ as an orthogonal operator on $L_2([0,1],\mathbb{R}^4)$ acting on $u\in L_2([0,1],\mathbb{R}^4)$ by pointwise multiplication:
$$
(Wu)(t)=W(t)u(t).
$$
The space $H_{0,0}^1$ is invariant under the action of  $W$. Then for any $Q\in T^2_{AGV}$
a tensor $W^\ast QW\in  T^2_{AGV}$ is defined by
$$
W^\ast QW(u,v)=Q(Wu,Wv),\;\;u,v\in H^1_{0,0}.
$$
\begin{definition}
The modified L\'evy trace associated with the curve $W\in C^1([0,1],SO(4))$ is a linear functional $\mathrm{tr}^W_L$ acting on $Q\in T^2_{AGV}$  by
$$
\mathrm{tr}^W_L Q=\mathrm{tr}_L^{AGV}(W^\ast QW).
$$ 
The modified  L\'evy Laplacian $\Delta_L^W$ associated with the curve $W\in C^1([0,1],SO(4))$ is the second order differential operator $D^{2,\mathrm{tr}^W_L}$.
\end{definition}

Where two normal subgroups $S^3_L$ and $S^3_R$ of $SO(4)$, which consist of real matrices
$$\begin{pmatrix}
a&-b&-c&-d\\
b&\;\,\, a&-d&\;\,\, c\\
c&\;\,\, d&\;\,\, a&-b\\
d&-c&\;\,\, b&\;\,\, a
\end{pmatrix}\text{ and }
\begin{pmatrix}
a&-b&-c&-d\\
b&\;\,\, a&\;\,\, d&-c\\
c&-d&\;\,\, a&\;\,\, b\\
d&\;\,\, c&-b&\;\,\, a
\end{pmatrix},
$$
where $a^2+b^2+c^2+d^2=1$, respectively. 
 The Lie algebras  $\mathrm{Lie}(S^3_L)$ and $\mathrm{Lie}(S^3_R)$ of the Lie groups  $S^3_L$ and $S^3_R$ consist of  matrices of the form   
$$\begin{pmatrix}
0&-b&-c&-d\\
b&\;\,\, 0&-d&\;\,\, c\\
c&\;\,\, d&\;\,\, 0&-b\\
d&-c&\;\,\, b&\;\,\, 0
\end{pmatrix}\text{ and }
\begin{pmatrix}
0&-b&-c&-d\\
b&\;\,\, 0&\;\,\, d&-c\\
c&-d&\;\,\, 0&\;\,\, b\\
d&\;\,\, c&-b&\;\,\, 0
\end{pmatrix},
$$
where $b,c,d\in \mathbb{R}$, respectively, So it holds that $$so(4)=\mathrm{Lie}(S^3_L)\oplus \mathrm{Lie}(S^3_R).$$

 Let the symbols  $P_L$ and $P_R$ denote the orthogonal projections in the Lie algebra $so(4)$ on its subalgebras  $\mathrm{Lie}(S^3_L)$  and $\mathrm{Lie}(S^3_R)$ respectively.
If $W\in C^1([0,1],SO(4))$ let ${\bf L}_W(t)=W^{-1}(t)\dot{W}(t)$. Then ${\bf L}_W\in C([0,1],so(4))$.
Let ${\bf L}_W^+(t)=P_L({\bf L}_W(t))$ and ${\bf L}_W^-(t)=P_R({\bf L}_W(t))$.
Due to the fact that $Q^S(t)$ is anti symmetric, it can be also considered as an element from the algebra $so(4)$.
Let $Q^S_+(t)=P_L(Q^S(t))$ and $Q^S_-(t)=P_R(Q^S(t))$.

It can be checked by direct computations that the following proposition holds (see~\cite{Volkov2020}).
\begin{proposition}
\label{Prop1}
Let $W\in C^1([0,1],SO(4))$.
It  holds
\begin{multline}
\mathrm{tr}^{W}_LQ=\int_0^1\mathrm{tr}Q^L(t)dt-\int_0^1\mathrm{tr}({\bf L}_W(t)Q^S(t))dt=\\
=\int_0^1\mathrm{tr}Q^L(t)dt-\int_0^1\mathrm{tr}({\bf L}_W^+(t)Q^S_+(t))dt-\int_0^1\mathrm{tr}({\bf L}_W^-(t)Q^S_-(t))dt.
\end{multline}
\end{proposition}

\begin{example}
\label{example1}
Let $\mathfrak{f}\in C^\infty(M,\mathbb{R})$. Let
$\mathfrak{L_{f}}\colon \Omega_m\to \mathbb{R}$
be defined by:
$$
\mathfrak{L_{f}}(\gamma)=\int_0^1\mathfrak{f}(\gamma(t))dt.
$$
The functional $\mathfrak{L_{f}}$ belongs to the domain of the L\'evy Laplacian $\Delta^{AGV}_L$ and the singular part of the second derivative of  $\mathfrak{L_{f}}$ vanishes.
One can show that
$$
 \Delta^{W}_L \mathfrak{L_{f}}(\gamma)=\int_0^1 \Delta_{(M,g)} \mathfrak{f}(\gamma(t))dt,
$$
 for any $W\in C^1([0,1],SO(4))$. Here $\Delta_{(M,g)}$ is the Laplace--Beltrami operator on the manifold $M$.
\end{example}

\subsection{Parallel transport and modified L\'evy Laplacian}
\label{Sec4c}

The parallel transport $U^A_{1,0}$ belongs to the domain of the L\'evy Laplacian (see~\cite{Volkov2020,LV2001}).
 We have
\begin{multline}
\label{der23}
<\widetilde{D}^2U^A_{1,0}(\gamma)u,v>=\\
=\int_0^1\int_0^1 K^V(\gamma;t,s)<u(t),v(s)>dtds+\\
+\int_0^1K^L(\gamma;t)<u(t),v(t)>dt+\\
+\frac 12\int_0^1K^S(\gamma;t)<\dot{u}(t),v(t)>dt+\frac 12\int_0^1K^S(\gamma;t)<\dot{v}(t),u(t)>dt,\;\;u,v\in H^1_{0,0},
\end{multline}
where the L\'evy kernel $K^L$ and the singular kernel $K^S$  have  the form
\begin{multline*}
K^L_{\mu\nu}(\gamma;t)=\frac 12U^A_{1,t}(\gamma)(-\nabla_{e_\mu(\gamma,t)} F(\gamma(t))<e_\nu(\gamma,t),\dot{\gamma}(t)>-\\-\nabla_{e_\nu(\gamma,t)} F(\gamma(t))<e_\mu(\gamma,t),\dot{\gamma}(t)>)U^A_{t,0}(\gamma),
\end{multline*}
and
$$
K^S_{\mu\nu}(\gamma;t)=U^A_{1,t}(\gamma)F(\gamma(t))<e_\mu(\gamma,t),e_\nu(\gamma,t)>U^A_{t,0}(\gamma).
$$

\begin{remark}
The form of the Volterra kernel in~(\ref{der23}) does not play a role  in the present paper. The exact expression of this kernel can been found in~\cite{Volkov2020}.
\end{remark}

The following theorem directly follows from Proposition~\ref{Prop1}.

\begin{theorem}
\label{Theorem1}
The value of the modified  L\'evy Laplacian $\Delta_L^W$ on  the  parallel transport is
\begin{multline}
\label{lllaplace00}
\Delta^W_LU^A_{1,0}(\gamma)=-
\int_0^1U^A_{1,t}(\gamma)D_A^\ast F(\gamma(t))\dot{\gamma}(t)U^A_{t,0}(\gamma)dt+\\
+\int_0^1U^A_{1,t}(\gamma)\mathrm{tr}({\bf L}_W(t)F(\gamma(t))U^A_{t,0}(\gamma)dt.
\end{multline}
\end{theorem}

\begin{remark}
Equality~(\ref{lllaplace00}) can been rewritten in the form
\begin{multline}
\label{lllaplace1}
\Delta^W_LU^A_{1,0}(\gamma)=-
\int_0^1U^A_{1,t}(\gamma)D_A^\ast F(\gamma(t))\dot{\gamma}(t)U^A_{t,0}(\gamma)dt+\\
+\int_0^1U^A_{1,t}(\gamma)\mathrm{tr}({\bf L}_W^+(t)F_+(\gamma(t))U^A_{t,0}(\gamma)dt+\\
+\int_0^1U^A_{1,t}(\gamma)\mathrm{tr}({\bf L}_W^-(t)F_-(\gamma(t))U^A_{t,0}(\gamma)dt.
\end{multline}
Now we can illustrate the relationship between the modified L\'evy Laplacian and instantons.
Let $F$ be an instanton. Then $F_+=0$ and $D_A^\ast F=0$. Hence, the first and the second terms in the right side 
of~(\ref{lllaplace1}) vanish. If additionally  $W\in C^1([0,1],S^3_L)$, then ${\bf L}_W^+=0$ for all $t\in[0,1]$ and 
 the right side of~(\ref{lllaplace1}) is zero.
\end{remark}

The first term in the right side of~(\ref{lllaplace00}) is invariant under reparameterization of the curve $\gamma$
but the second  term is not. Therefore, the following lemma is true.
\begin{lemma}
\label{lemma1}
Let $W\in C^1([0,1],SO(4))$.
If the parallel transport  $U^A_{1,0}$ is a solution of the Laplace equation for the modified  L\'evy Laplacian   $\Delta_L^{W}$:
$$
 \Delta_L^{W} U^A_{1,0}=0,
$$
  then the connection $A$ satisfies the Yang--Mills equations and for any $\gamma\in \Omega_m$ the following holds
  $$
 \int_0^1 U^A_{1,t}(\gamma)\mathrm{tr}({\bf L}_W(t)F(\gamma(t))U^A_{t,0}(\gamma)dt=0.
  $$
\end{lemma}

The  complete  proof of the  lemma can been found in~\cite{Volkov2020}.

\begin{remark}
In paper~\cite{Volkov2019a} by the author, the covariant definition of the L\'evy Laplacian on manifolds was introduced.
A natural question arises whether it is
 possible to introduce the covariant analog of the modified L\'evy Laplacian that is invariant under Weyl transformation  and is invariant under reparametrization of the curves.
\end{remark}

\section{Main theorem}
\label{Sec4}
In this section we give the formulation and the proof of the main theorem.

We can assume that $\{e_1,e_2,e_3,e_4\}$ is a right-handed basis in $T_mM$ without loss of generality.
Let us introduce bivectors in $\Lambda^2(T_mM)$:
$$
v_1^{\pm}=\frac 1{\sqrt{2}}(e_1\wedge e_2\pm e_3\wedge e_4),
$$
$$
v_2^{\pm}=\frac 1{\sqrt{2}}(e_1\wedge e_3\mp e_2\wedge e_4),
$$
$$
v_3^{\pm}=\frac 1{\sqrt{2}}(e_1\wedge e_4\pm e_2\wedge e_3).
$$
Let $v^{\pm}_i(\gamma,\cdot)$ be the parallel transport of $v^{\pm}_i$ ($i\in\{1,2,3\}$) along $\gamma\in \Omega_m$ by the Riemannian connection. Then $\{v_1^{+}(\gamma,t),v_2^{+}(\gamma,t),v_3^{+}(\gamma,t)\}$ and $\{v_1^{-}(\gamma,t),v_2^{-}(\gamma,t),v_3^{-}(\gamma,t)\}$ are  orthonormal bases in $\Lambda^2_+(T_{\gamma(t)}M)$ and
 $\Lambda^2_-(T_{\gamma(t)}M)$ respectively for any $t\in [0,1]$.

 We identify an element from $so(4)$ with the action on $T_mM$.
  Let in the basis $\{e_1,e_2,e_3,e_4\}$
  the operators ${\bf L}_W^+(t)$   and  ${\bf L}_W^-(t)$ have 
  matrices
  \begin{equation}
\label{BL}
{\bf L}_W^+(t)=
\begin{pmatrix}
0&-\omega_{W1}^+(t)&-\omega_{W2}^{+}(t)&-\omega_{W3}^{+}(t)\\
\omega_{W1}^+(t)&\;\,\, 0&-\omega_{W3}^{+}(t)&\;\,\, \omega_{W2}^{+}(t)\\
\omega_{W2}^{+}(t)&\;\,\, \omega_{W3}^{+}(t)&\;\,\, 0&-\omega_{W1}^{+}(t)\\
\omega_{W3}^{+}(t)&-\omega_{W2}^{+}(t)&\;\,\, \omega_{W1}^{+}(t)&\;\,\, 0
\end{pmatrix}
\end{equation}
and
 \begin{equation}
{\bf L}_W^-(t)=
\begin{pmatrix}
0&-\omega_{W1}^{-}(t)&-\omega_{W2}^{-}(t)&-\omega_{W3}^{-}(t)\\
\omega_{W1}^{-}(t)&\;\,\, 0&\;\,\, \omega_{W3}^{-}(t)&-\omega_{W2}^{-}(t)\\
\omega_{W2}^{-}(t)&-\omega_{W3}^{-}(t)&\;\,\, 0&\;\,\, \omega_{W1}^{-}(t)\\
\omega_{W3}^{-}(t)&\;\,\, \omega_{W2}^{-}(t)&-\omega_{W1}^{-}(t)&\;\,\, 0
\end{pmatrix}
\end{equation}
respectively.

Let for any $\gamma\in \Omega_m$ and $t,r\in [0,1]$
bivectors $v_{W}^-(\gamma,t,r)\in \Lambda^2_-(T_{\gamma(t)}M))$ and
$v_W^+(\gamma,t,r)\in \Lambda^2_+(T_{\gamma(t)}M))$ be defined by
$$
v_{W}^\pm(\gamma,t,r)=\omega_{W1}^\pm(r)v_1^\pm(\gamma,t)+\omega_{W2}^{\pm}(r)v_2^\pm(\gamma,t)+\omega_{W3}^{\pm}(r)v_3^\pm(\gamma,t).
$$

Let $\alpha_{Wi}^{\pm}=\int_0^1\omega_{Wi}^\pm(r)dr$ ($i\in \{1,2,3\}$) and
 $$w_W^\pm(\gamma,t)=\alpha_{W1}^\pm v_1^\pm(\gamma,t)+\alpha_{W2}^\pm v_2^\pm(\gamma,t)+\alpha_{W3}^\pm v_3^\pm(\gamma,t).$$

It is easy to see that $$w_W^\pm(\gamma,t)=\int_0^1v^\pm_W(\gamma,t,r)dr.$$

We will use the following notations for   all $\gamma\in \Omega_m$ and $t,r\in[0,1]$:
\begin{equation*}
L(\gamma,t)=U^A_{t,0}(\gamma)^{-1}F(\gamma(t))U^A_{t,0}(\gamma),
\end{equation*}
\begin{equation*}
L_{\pm}(\gamma,t)=U^A_{t,0}(\gamma)^{-1}F_{\pm}(\gamma(t))U^A_{t,0}(\gamma),
\end{equation*}
$$
L^W(\gamma,t,r)= \mathrm{tr}(\dot{W}(r)W^{-1}(r)L(\gamma,t)).
$$
It can been obtained by direct computations that
\begin{multline*}
L^W(\gamma,t,r)=\mathrm{tr}(\dot{W}(r)W^{-1}(r)L(\gamma,t))=\mathrm{tr}({\bf L}_W^-(r)L_-(\gamma,t))+\mathrm{tr}({\bf L}_W^+(r)L_+(\gamma,t))=\\
=L_-(\gamma,t)<v_{W}^-(\gamma,t,r)>+L_+(\gamma,t)<v_W^+(\gamma,t,r)>.
\end{multline*}

Due to Theorem~\ref{Theorem1} we have
\begin{equation*}
\Delta^W_LU^A_{1,0}(\gamma)=-
\int_0^1U^A_{1,t}(\gamma)D_A^\ast F(\gamma(t))\dot{\gamma}(t)U^A_{t,0}(\gamma)dt
+U^A_{1,0}(\gamma)\int_0^1L^W(\gamma,t,t)dt.
\end{equation*}
Lemma~\ref{lemma1} implies that if $ \Delta_L^{W} U^A_{1,0}=0$
then $\int_0^1L^W(\gamma,t,t)dt$ for any $\gamma \in \Omega_m$.

\begin{lemma}
\label{lemma2}
Let $W\in C^1([0,1],S^3_L)$.
If the parallel transport  $U^A_{1,0}$ is a solution of the Laplace equation for the L\'evy Laplacian   $\Delta_L^{W}$:
$ \Delta_L^{W} U^A_{1,0}=0,$
then for any $\gamma\in \Omega_m$ the following holds
  $$
 F_+(\gamma(1))<w_W^+(\gamma,1)>=0.
  $$
\end{lemma}
\begin{proof}

If  $W\in C^1([0,1],S^3_L)$, then ${\bf L}^-_W=0$ and
\begin{equation}
L^W(\gamma,t,r)=L_+(\gamma,t)<v_{W}^+(\gamma,t,r)>.
\end{equation}

For any  $r\in (0,1]$ let introduce $\gamma_r\in\Omega_m$ by 
$$
\gamma_{r}(t)=\begin{cases}
\gamma(t/r), &\text{if $0\leq t\leq r$,}\\
\gamma(1),  &\text{if $r<t\leq1$.}\\
\end{cases}
$$

The parallel transport is invariant under reparameterization of the curve $\gamma$ (see~\cite{Volkov2020}). Hence,
$$
L^W(\gamma_r,t,t)=\begin{cases}L^W(\gamma,t/r,t)
, &\text{if $0\leq t\leq r$,}\\
L^W(\gamma,1,t),  &\text{if $r<t\leq1$.}\\
\end{cases}
$$

Then
\begin{equation}
\label{eq}
   \int_0^1 L^W(\gamma_r,t,t)dt
   =\int_0^rL^W(\gamma,t/r,t)dt
   +\int_r^1L_+(\gamma(1))<v^+_W(\gamma,1,t)>dt.
  \end{equation}
There exist $C>0$ such that
  $$
\sup_{t\in[0,1],r\in[0,1]} \| L^W(\gamma,t,r)\|<C.
  $$
 Hence
 $$
 \|\int_0^rL^W(\gamma,t/r,t)dt\|\leq Cr
 $$
and
$$
\lim_{r\to 0+}\int_0^rL^W(\gamma,t/r,t)dt=0.
$$

The last equality and~(\ref{eq}) together imply 
\begin{multline}
\label{eq2}
\lim_{r\to 0+}\int_0^1L^W(\gamma_r,t,t)dt=\\=\int_0^1L_+(\gamma(1))<v^+_W(\gamma,1,t)>dt=L_+(\gamma(1))<w^+_W(\gamma,1)>.
\end{multline}
If $ \Delta_L^{W} U^A_{1,0}=0$
then $\int_0^1L^W(\gamma_r,t,t)dt$ for any $\gamma \in \Omega_m$ and $r\in(0,1]$. Hence, it follows from~(\ref{eq2}) that
$$
L_+(\gamma(1))<w^+_W(\gamma,1)>=0.
$$
So we have $F_+(\gamma(1))<w_W^+(\gamma,1)>=0$ for any $\gamma\in \Omega_m$.
\end{proof}

In the case
$\mathrm{Hol}_m^0(\Lambda^2_+(TM))\cong U(1)$, 
we can assume without loss of generality that $v_1^{+}$ is invariant under any action from this group.
Let $V_1=\mathrm{span}\{v_1^{+}\}$ and $V_2=\mathrm{span}\mathrm\{v_2^{+},v_3^{+}\}$. Then the spaces $V_1$  and $V_2$ are   invariant under   any action from $\mathrm{Hol}_m^0(\Lambda^2_+(TM))$.
For any $\gamma\in \Omega_m$ bivector  $v_1^{+}(\gamma,1)$ is invariant  under any action from $\mathrm{Hol}_{\gamma(1)}^0(\Lambda^2_+(TM))$
and  the orbit
of $v_2^{+}(\gamma,1)$ and $v_3^{+}(\gamma,1)$ coincides with 
$$
\mathrm{Orb}(v_2^{+}(\gamma,1))=\mathrm{Orb}(v_3^{+}(\gamma,1))=\{v_2^{+}(\gamma,1)\cos\theta+v_3^{+}(\gamma,1)\sin\theta\colon\theta\in [0,2\pi)\}.
$$
Also we can canonically identify  $\mathrm{Lie}(S^3_L)$ and  $\Lambda^2_+(T_mM)$
by
$$
\begin{pmatrix}
0&-b&-c&-d\\
b&\;\,\, 0&-d&\;\,\, c\\
c&\;\,\, d&\;\,\, 0&-b\\
d&-c&\;\,\, b&\;\,\, 0
\end{pmatrix}
\leftrightarrow bv_{1}^{+}+cv_{2}^{+}+dv_{3}^{+}.
$$
So we can define orthogonal projections $\mathrm{pr}_{V_1}$ and $\mathrm{pr}_{V_2}$ in $\mathrm{Lie}(S^3_L)$ on the spaces $V_1$ and $V_2$ respectively.

Now we can formulate and prove the main theorem.

\begin{theorem}
\label{MainTheor}
Let $M$ be an orientable  compact Riemannian 4-manifold.
\begin{enumerate}
\item If  $\mathrm{Hol}_m^0(\Lambda^2_+(TM))\cong SO(3)$ and $W\in C^1([0,1],S^3_L)$ such that $\int_0^1{\bf L}_W(t)dt\neq 0$;
\item If  $\mathrm{Hol}_m^0(\Lambda^2_+(TM))\cong U(1)$ and $W\in C^1([0,1],S^3_L)$ such that $\mathrm{pr}_{V_1}(\int_0^1{\bf L}_W(t)dt)\neq 0$ and
$\mathrm{pr}_{V_2}(\int_0^1{\bf L}_W(t)dt)\neq 0$
\end{enumerate} 
 the following two assertions are equivalent:
\begin{enumerate}
\item  a connection $A$ is a solution of anti-self-duality equations : $F=-\ast F$,
\item  the parallel transport  $U^A_{1,0}$ is a solution of the equation:
$
 \Delta_L^{W} U^A_{1,0}=0.$
\end{enumerate}
\end{theorem}
\begin{proof}
Let $\Delta_L^{W} U^A_{1,0}=0$. 
Fix any $x\in M$.  Fix an arbitrary $\sigma\in \Omega_{m,x}$.
Then the set $\{w^+_W(\gamma,1)\colon \gamma\in \Omega_{m,x}\}$ coincides with the orbit $\mathrm{Orb}(w^{+}_W(\sigma,1))$ of $w^+_W(\sigma,1)$
 under the action of the group $\mathrm{Hol}^0_x(\Lambda^2_+(TM))$. Then Lemma~\ref{lemma2} implies that $F_+(x)<w>=0$ for any $w\in \mathrm{Orb}(w^{+}_W(\sigma,1))$. Now we can consider two cases.
 
\begin{enumerate}
\item 
Let   $\mathrm{Hol}^0_m(\Lambda^2_+(TM))\cong SO(3)$ and $\int_0^1{\bf L}_W(t)dt\neq0$.  The last  condition is equal to  $\sum_{i=1}^3(\alpha^{+}_{Wi})^2\neq 0$.
In this case, $w^+_W(\sigma,1)\neq 0$. Due to $\mathrm{Hol}_x^0(\Lambda^2_+(TM))=SO(3)$,
the orbit $\mathrm{Orb}(w^{+}_W(\sigma,1))$ is a 2-sphere. There are three linearly independent vectors in  $\mathrm{Orb}(w^{+}_W(\sigma,1))$. 

\item 
Let   $\mathrm{Hol}_m^0(\Lambda^2_+(TM))\cong U(1)$ and  $\mathrm{pr}_{V_1}(\int_0^1{\bf L}_W(t)dt)\neq 0$ and
$\mathrm{pr}_{V_2}(\int_0^1{\bf L}_W(t)dt)\neq 0$.
In this case,
the  orbit of $w^{+}_W(\sigma,1)$ has a form
\begin{multline*}
\mathrm{Orb}(w^{+}_W(\sigma,1))=\\
=\{\alpha^{+}_{W1}v_1^{+}(\sigma,1)+\alpha^{+}_{W2}v_2^{+}(\sigma,1)\cos\theta+\alpha^{+}_{W2}v_3^{+}(\sigma,1)\sin\theta\colon\theta\in [0,2\pi)\}.
\end{multline*}
The conditions 
$\mathrm{pr}_{V_1}(\int_0^1{\bf L}_W(t)dt)\neq 0$ and
$\mathrm{pr}_{V_2}(\int_0^1{\bf L}_W(t)dt)\neq 0$ means that $\alpha^{+}_{W1}\neq 0$ and $(\alpha^{+}_{W2})^2+(\alpha^{+}_{W3})^2\neq 0$.
In this case,  there are three linearly independent vectors in  $\mathrm{Orb}(w^{+}_W(\sigma,1))$.
\end{enumerate}
In the both cases, the equality  $F_+(x)<w>=0$ for any $w\in \mathrm{Orb}(w^{+}_W(\sigma,1))$ imply $F_+(x)=0$. Hence,
$A$ is an anti-self-dual connection.
The other side of the theorem is trivial.
\end{proof}

\begin{remark}
Let $W\in C^1([0,1],S^3_R)$.
If we change the orientation of the manifold $M$ in the Theorem~\ref{MainTheor} we obtain the similar equivalence of the the Yang--Mills self-duality  equations and   the Laplace equation for the modified L\'evy Laplacian  $\Delta_L^{W}$.
\end{remark}

\begin{remark}
If the simply connected base manifold
 $(M,g)$ is a Calabi--Yau manifold then Theorem~\ref{MainTheor} does not hold. As it was mentioned in the introduction there exists the Weyl transformation $g\to \overline{g}$ such that
for $(M,\overline{g})$ Theorem~\ref{MainTheor} holds.
\end{remark}

\section*{Conclusion}

In this paper, we have shown that the Yang--Mills self-duality and anti-self-duality  equations on a connection  in the vector bundle over the  Riemannian 4-manifold can been reformulated as a linear differential equation on the parallel transport generated by the connection.
The L\'evy Laplacian is an infinite dimensional differential operator  on the space of sections in the vector bundle over the space of $H^1$-curves in the Riemannian manifold. It can been defined as an integral functional generated by the special form of the second order derivative. The equivalence of the Laplace equation for the L\'evy Laplacian and the Yang--Mills equations was previously known (see~\cite{AGV1994,LV2001}). The L\'evy Laplacian is not rotation invariant and the infinite dimensional rotation $W\in C^1([0,1],SO(4))$ generates the modified  L\'evy Laplacian $\Delta_L^W$. The group of 4-dimensional rotations $SO(4)$ has two normal subgroups $S^3_L$ and $S^3_R$.
In the previous works~\cite{VolkovLLI,Volkov2020},  a connection between Laplace equations for the modified L\'evy Laplacians $\Delta_L^W$
 for $W\in C^1([0,1],S^3_L)$ (for $W\in C^1([0,1],S^3_R)$)  and the Yang--Mills anti-self-duality (self-duality) equations was discovered.
In this paper,  we have strengthened the results of  these  works. Under the assumption that  the restricted  holonomy group of the bundle of self-dual 2-forms of the 4-manifold is nontrivial,  we have  found the sufficient conditions on  the rotation $W\in C^1([0,1],S^3_L)$ ($W\in C^1([0,1],S^3_R)$)
such that the following holds.  A connection is an instanton (antiinstanton)  if and only if  the parallel transport is a solution  to the  Laplace equation for the modified L\'evy Laplacian $\Delta_L^W$.

\section*{Acknowledgments}

The author would like to express his deep gratitude to L.~Accardi, O.~G.~Smolyanov and I.~V.~Volovich for helpful discussions.

\end{document}